\documentclass[11pt]{amsart}
\usepackage[utf8]{inputenc}
\usepackage{verbatim}
\usepackage{amssymb,amsmath,euscript}
\usepackage{xcolor}

\newtheorem{theorem}{Theorem}[section]
\newtheorem{lemma}[theorem]{Lemma}

\newtheorem{conjecture}[theorem]{Conjecture}

\theoremstyle{definition}
\newtheorem{definition}[theorem]{Definition}

\theoremstyle{remark}
\newtheorem{remark}[theorem]{Remark}

\newcommand{\Fq}{\mathbb{F}_q}
\newcommand{\Glm}{G_{\ell, m}}

\newcommand{\blV}{\bigwedge^\ell V}
\newcommand{\bmlV}{\bigwedge^{m-\ell} V}

\DeclareMathOperator{\wt}{wt}
\DeclareMathOperator{\codim}{codim}
\DeclareMathOperator{\Ev}{Ev}
\DeclareMathOperator{\rank}{rank}
\newcommand{\Oalm}{\Omega_\alpha(\ell,m)}
\newcommand{\Otlm}{\Omega_\theta(\ell,m)}
\newcommand{\Calm}{C_\alpha(\ell,m)}
\newcommand{\Ctlm}{C_\theta(\ell,m)}

\begin{document}	

\title[]{Minimum Schubert Codewords and Second-Minimum Grassmann Codewords}
\author{Muskan Khaneja}
\address{Department of Mathematics,\newline \indent
	IIT Jammu,\newline \indent
	Jammu \& Kashmir, 181221.}
\email{muskanhcps@gmail.com}
\author{Prasant Singh}
\address{Department of Mathematics,\newline \indent
	IIT Jammu,\newline \indent
	Jammu \& Kashmir, 181221.}
\email{psinghprasant@gmail.com}
\date{\today}

\begin{abstract}
In this paper, we give a classification of the minimum weight codewords of Schubert codes $\Calm$ by settling the conjecture \cite[Conj. 5.6]{ghorpade_singh} for all values of $q$ and all $\alpha$. We use this classification to prove that a codeword of the Grassmann code $C(\ell, m)$ has the second minimum weight if and only if it is indexed by an element of $\bigwedge^{m-\ell}V$ that can be written as the product of a decomposable $(m-\ell-2)$-vector and an alternating $2$-vector of rank $4$. Finally, we give an enumeration of the second minimum weight codewords of the Grassmann code.

\end{abstract}

\maketitle


\section{Introduction}
Let $q$ be a prime power and let $\Fq$ denote the finite field with $q$ elements. Let $\ell, m$ be positive integers satisfying $\ell\leq m$ and let $V$ be an $m$-dimensional vector space over the field $\Fq.$ Let $G(\ell, V)$ denote the set of all $\ell$-dimensional subspaces of $V$. The Pl\"ucker map embeds $G(\ell, V)$ into the projective space $\mathbb{P}(\bigwedge^\ell V)$ nondegenerately. It can be shown that, geometrically, this embedding depends only on the positive integers $\ell$ and $m$. Therefore, we denote this set by $G(\ell, V)$ as well as by $G_{\ell, m}$ and call it the Grassmannian of $\ell$-planes of $V$. As a subset of the projective space $\mathbb{P}(\bigwedge^\ell V)$, the Grassmannian can be viewed as a projective system \cite[Ch. 1]{TVN1}. The linear code obtained from the projective system $G_{\ell, m}$ is known as the Grassmann code and is denoted by $C(\ell, m)$. Grassmann codes and their properties have been an active area of research since they were first introduced in a series of papers by Ryan and Ryan over the binary field \cite{C_Ryan1, C_Ryan2, cryan_kryan} and later by Nogin \cite{nogin_minimum_distance} over a general finite field. For example, Ryan and Ryan in the binary case and Nogin in general proved that the minimum distance of the Grassmann code $C(\ell, m)$ is $q^{\ell(m-\ell)}$ and that the minimum weight codewords of $C(\ell, m)$ are indexed by completely decomposable elements of $\bigwedge^{m-\ell}V$. Further, Nogin \cite[Thm. 3]{nogin_minimum_distance} determined the weight spectrum of the Grassmann code $C(2, m)$. He also determined \cite[Thm. 4]{nogin_minimum_distance} the first $\mu= \max\{\ell+1, m-\ell+1\}$ higher weights of the Grassmann code $C(\ell, m)$ and showed that they are given by the Griesmer-Wei bounds. Continuing this work, Nogin \cite{nogin_g36} determined the weight spectrum of $C(3, 6)$, while Kaipa and Pillai \cite{kaipa_g3_7} determined the weight spectrum of $C(3, 7)$. Ghorpade and Kaipa determined the full automorphism group of $C(\ell, m)$. Further, Ghorpade, Patil and Pillai determined the $(\mu+1)^{\it th}$ higher weight of the Grassmann code $C(2, m)$. It is worth pointing out that, in the computation of $d_{\mu+1}$, they used not only the second minimum weight of the code $C(2, m)$ but also the structure of the codewords having the second minimum weight. This illustrates that the Grassmann codes have been studied extensively by different groups of mathematicians over the last three decades.

Schubert varieties are special subvarieties of the Grassmannian $G_{\ell, m}$ obtained by intersecting $G_{\ell, m}$ with certain systematic coordinate hyperplanes of $\mathbb{P}(\bigwedge^\ell V)$. To be precise, for any $\ell$-tuple $\alpha=(\alpha_1,\alpha_2,\dots, \alpha_\ell)$ of a strictly monotonic increasing sequence of positive integers between $1$ and $m$, let $\Oalm$ be the subvariety of the Grassmannian obtained by intersecting $G_{\ell, m}$ with the Pl\"ucker coordinate hyperplanes $P_\beta=0$ where $\beta\nleq\alpha$ under the Bruhat order. The restriction of the Pl\"ucker map to $\Oalm$ embeds the Schubert variety $\Oalm$ into the projective space $\mathbb{P}(\bigwedge^\ell V)$, but this embedding is not non-degenerate in general. Nevertheless, one could consider the smallest projective subspace of $\mathbb{P}(\bigwedge^\ell V)$ where this variety lies and consider the corresponding projective system to obtain the linear code. The linear code obtained from $\Oalm$ in this way is known as the Schubert code and is denoted by $\Calm$. Schubert codes were introduced by Ghorpade and Lachaud in \cite{ghorpade_and_lachaud}. They proved that the minimum distance of $C_{\alpha}(\ell,m)$ satisfies

$$
d(C_{\alpha}(\ell,m)) \leq q^{\delta(\alpha)},
\quad \text{where} \quad 
\delta(\alpha):=\sum_{i=1}^{\ell}(\alpha_{i}-i).
$$

They also conjectured in \cite{ghorpade_and_lachaud} that the inequality above is, in fact, an equality. Ghorpade and Singh \cite{ghorpade_singh} referred to this as the Minimum Distance Conjecture (MDC). For $\ell=2$, the MDC was first proved by Chen \cite{Chen}, and independently by Guerra and Vincenti \cite{Guerra_Vincenti}. The conjecture was later established in full generality by Xiang \cite{xiang}. Ghorpade and Singh \cite{ghorpade_singh} also gave an alternative, coordinate-free proof of the MDC. In addition, they proposed a characterization of the minimum weight codewords of Schubert codes in terms of Schubert decomposable elements of $\bigwedge^{m-\ell}V$ \cite[Conj. 5.6]{ghorpade_singh} and proved several aspects of this conjecture. For example, they proved that the codewords of $\Calm$ indexed by Schubert decomposable elements of $\bigwedge^{m-\ell}V$ are minimum weight codewords, and that if a minimum weight codeword is indexed by some decomposable element of $\bigwedge^{m-\ell}V$, then that element must be Schubert decomposable. Furthermore, they showed that every minimum weight codeword of $C_{\alpha}(\ell,m)$ can be indexed by a decomposable element of $\bigwedge^{m-\ell}V$ whenever $\ell=2$ or $\alpha$ is completely non-consecutive. But the conjecture is still open in general. However, recently, Datta and Dutta \cite{datta2026characterizationminimumweightcodewords}, using combinatorial decomposition, proved that every minimum weight codeword can be indexed by a decomposable element of $\bigwedge^{m-\ell}V$ for all but finitely many values of $q$. Specifically, they established this result for

$$
q > q_0(\ell)= \dfrac{2^{\dfrac{1}{\ell-1}}}{2^{\dfrac{1}{\ell-1}}-1}.
$$

It follows from \cite{datta2026characterizationminimumweightcodewords} that $q_0(\ell)$ is an increasing function of $\ell$; nevertheless, the conjecture remains unsolved for infinitely many pairs $(q,\ell)$.

In this paper, we extend the coordinate-free methods developed by Ghorpade and Singh \cite{ghorpade_singh} and prove that every minimum weight codeword of $C_{\alpha}(\ell,m)$ can be indexed by a decomposable element of $\bigwedge^{m-\ell}V$ for arbitrary $\ell$ and $q$. In other words, we settle the conjecture \cite[Conj. 5.6]{ghorpade_singh}. Combined with the results in \cite{ghorpade_singh}, we now have a complete classification of the minimum weight codewords of Schubert codes in terms of Schubert decomposable elements.

The problem of computing the weight spectrum of Grassmann codes and the generalized Hamming weights of Grassmann codes $C(\ell, m)$ is considerably difficult. Recently, Datta and Dutta \cite{datta2026secondminimumweightgrassmann} took a step in the direction of computing the weight spectrum of $C(\ell, m)$ and determined the second minimum weight of $C(\ell,m)$ using a new combinatorial decomposition of the Grassmannian. As mentioned earlier, Ghorpade, Patil and Pillai \cite{ghorpade_decomposable} used the second minimum weight of $C(2, m)$ and the structure of such codewords in computing the higher weight $d_{\mu+1}$. Thus, it is natural to ask for a classification of the second minimum weight codewords of $C(\ell, m)$, as this is an interesting problem from both a coding-theoretic and a geometric point of view. Such a classification may also be useful in computing the higher weights of Grassmann codes.

In their article, Datta and Dutta \cite{datta2026secondminimumweightgrassmann} also observed that a complete characterization of the second minimum-weight codewords of $C(\ell, m)$ requires a classification of the minimum-weight codewords of Schubert codes. Using the classification of the minimum weight codewords of $\Calm$ established in this paper, we prove that the elements $f \in \bigwedge^{m-\ell} V$ indexing the second minimum weight codewords of $C(\ell,m)$ are precisely of the form

$$
f=v_1 \wedge \cdots \wedge v_{m-\ell-2} \wedge \xi
$$

where $v_1, \ldots, v_{m-\ell-2} \in V$ are linearly independent and $\xi \in \bigwedge ^2 V$ is an alternating 2-vector of rank 4.

The paper is organized as follows. In Section 2, we recall the necessary background on Grassmannians, Schubert varieties, and Schubert codes. In Section 3, we classify the minimum weight codewords of Schubert codes. In Section 4, we use this classification to characterize the second minimum weight codewords of Grassmann codes and determine the number of such codewords.


\section{Preliminaries} \label{sec2}
In this section, we provide a brief introduction to the construction of Grassmann and Schubert codes and collect some results for these codes that we will use in this article. Grassmann codes and Schubert codes are algebraic geometric codes obtained from the projective systems of the sets of $\Fq$-rational points of Grassmannian and of Schubert varieties in the Grassmannian \cite[Ch. 11.5]{TVN2}. To be precise, let $\ell, m$ be positive integers satisfying $\ell\leq m$, and as before let $V$ be an $m$-dimensional vector space over $\Fq$. The set of $\Fq$-rational points of the Grassmann variety is defined and denoted by 
$$
G(\ell, V)=\{ L\subseteq V: L\text{ is a linear subspace and}\dim L=\ell\}.
$$
As in this article, we are going to deal with the set of $\Fq$-rational points of the Grassmann variety, and by abuse of language, we will call this set $G(\ell, V)$ the Grassmannian or Grassmann variety. If there is no confusion about the ambient vector space $V$, we use the notation $\Glm$ to denote the Grassmannian $G(\ell, V)$. It is well known that the Grassmannian $\Glm$ is embedded nondegenerately into the projective space $\mathbb{P}(\bigwedge^\ell V)$ via the Pl\"ucker map. We may consider $\Glm\subseteq \mathbb{P}(\bigwedge^\ell V)$ via the Pl\"ucker map. Further, fixing a basis $\{e_1, \dots, e_m\}$ of $V$ and by identifying $e_1\wedge\dots\wedge e_m=1$, we may consider $\bigwedge^{m-\ell} V$ as the dual of the  vector space $\bigwedge^\ell V$. Now, let $\{L_1,\dots, L_n\}$ be an enumeration of the points of $\Glm$ and let $P_1, \dots, P_n$ be their representatives in the space $\blV$ . Consider the map 
 \begin{align}
 \label{eq:defGC}
   \Ev: \bmlV &\to \Fq^{\,n}   \nonumber\\
   f&\mapsto c_f=(f(P_1),\dots, f(P_n)).
 \end{align}
 The evaluation map is linear and injective. The image of this map is called Grassmann code, and we denote it by $C(\ell, m)$. It has been shown by Ryan and Ryan \cite{C_Ryan1, C_Ryan2, cryan_kryan} over binary field and by Nogin \cite{nogin_minimum_distance} over a general finite field that $C(\ell, m)$ is an $[n, k, d]$-code, where 
 \begin{equation}
     \label{eq:parametersGC}
     n= \left[\begin{matrix} m \\ \ell \end{matrix}\right]_q,\;    k=\binom{m}{\ell},\;  \text{ and } d=q^\delta,
 \end{equation}
where $\left[\begin{matrix} m \\ \ell \end{matrix}\right]_q$ denotes the Guassian binomial and $\delta=\ell(m-\ell)$. Further, they also showed that a nonzero codeword of $C(\ell, m)$ is of minimum weight if and only if it is obtained by evaluating a decomposable element of $\blV.$

Let $I(\ell,m)$ be the set of all strictly increasing $\ell$-tuples of positive integers $\{1, 2,\dots, m\}$, i.e., 
$$
I(\ell,m)=\{ \alpha=(\alpha_1, \ldots, \alpha_\ell)\in \mathbb{Z}^\ell : 1 \leq \alpha_1 < \cdots < \alpha_\ell \leq m\}.
$$
We consider the set $I(\ell, m)$ under the partial order defined as: for $\alpha, \beta\in I(\ell, m)$ 
$$
\alpha\leq \beta \iff \alpha_i\leq \beta_i\text{ for }1\leq i\leq \ell. 
$$
The partial order defined above is well known as the \emph{Bruhat order}. Further, for $\alpha\in I(\ell, m)$ we set $\delta(\alpha)= \sum_{i=1}^\ell (\alpha_i-i).$

For a given $\alpha=(\alpha_1, \ldots, \alpha_\ell)\in I(\ell,m)$, let $A_1 \subset \cdots \subset A_\ell$  be a partial flag of subspaces of $V$ satisfying $\dim A_i=\alpha_i$ for $1\leq i\leq \ell$. The set of $\Fq$-rational points of the  Schubert variety corresponding to the dimension sequence $\alpha$ is defined and denoted by 
$$
\Omega_\alpha(\ell,m)=
\left\{
L \in G_{\ell,m}:
\dim (L \cap A_i) \geq i,\;
1 \leq i \leq \ell
\right\}.
$$
Again, by abuse of notation, we call the set $\Oalm$ the Schubert variety in the Grassmannian $\Glm$ corresponding to the dimension sequence $\alpha$. Note that if we take $\alpha$ to be consecutive, i.e.,  $\alpha=(m-\ell+1, m-\ell+2,\dots, m)$, then $\Oalm=\Glm$ and $\delta(\alpha)= \ell(m-\ell)$. The restriction of the Pl\"ucker map to Schubert varieties $\Oalm$ gives an embedding of $\Oalm$ into the projective space $\mathbb{P}(\bigwedge^\ell V)$ but this embedding is degenerate except when $\Oalm$ is the full Grassmannian $\Glm$. By taking the smallest projective subspace of $\mathbb{P}(\bigwedge^\ell V)$ containing $\Oalm$, we may consider this as a projective system and the Schubert code is the linear code obtained from this projective system \cite[Ch.  1]{TVN1}. In other words, one can consider the evaluation map defined in equation \eqref{eq:defGC} and restrict the map by evaluating the functions $f\in\bmlV$ on points that are in Schubert variety $\Oalm$. Since the embedding of $\Oalm$ into the projective space $\mathbb{P}(\bigwedge^\ell V)$ is degenerate in general, the evaluation map obtained by restricting \eqref{eq:defGC} is not injective. Nevertheless, the image is still a (degenerate) subspace of $\Fq^{\,n}$. The image of the restriction map described above is called the Schubert code and is denoted by $\Calm$. One could also realize the Schubert code as the code obtained from puncturing the Grassmann code $C(\ell, m)$ on the points $\Glm\setminus \Oalm.$ Further, it is worth mentioning that, \emph{since the restriction of the evaluation map is not injective, a codeword $c\in \Calm$ can be indexed by several $f\in \bmlV$}. It has now been established that the Schubert code $\Calm$ is an $[n_\alpha, k_\alpha, d_\alpha]$ linear code with 

\begin{equation}
    \label{eq:parametersSC}
    n_\alpha = \sum_{\beta\leq \alpha}q^{\delta(\beta)}, \quad 
k_\alpha= \det_{1 \leq i,j \leq \ell}\left( \left( 
\begin{matrix}
\alpha_j-j+1 \\ i-j+1 
\end{matrix} \right) \right), \quad{ and }\quad d_\alpha=q^{\delta(\alpha)}.
\end{equation}

Following the notations from \cite{ghorpade_singh}, it is natural to partition the sequence $\alpha$ into maximal blocks of consecutive integers as:
$$
\alpha=(\alpha_1,\ldots,\alpha_{p_1},
\alpha_{p_1+1},\ldots,\alpha_{p_2},
\ldots,
\alpha_{p_u+1},\ldots,\alpha_\ell),
$$
where the entries within each block are consecutive integers and
$$
\alpha_{p_j+1}-\alpha_{p_j} \geq 2,\qquad 1\le j\le u.
$$

For convention, define $p_0=0$ and $p_{u+1}=\ell$. Note that the integers $u, p_1, \ldots, p_u$ are uniquely determined by $\alpha$. For example, if $\ell=6$ and $\alpha=(1,2,4,5,6,8)$, then $u=2$ and $(p_1,p_2)=(2,5)$. The Schubert variety $\Oalm$ can now be redefined as 
$$
\Oalm=\{L\in G(\ell, V): \dim(L\cap A_{p_i})\geq i \text{ for }i= 1, 2, \dots, u\}. 
$$

Recall that  for any $f\in\blV$, the set of annihilators of $f$ is denoted by $V_f$ and it is defined as
$$
V_f=\{v\in V: f\wedge x= 0\}.
$$
It is well known that $V_f$ is a subspace of $V$ and if $f$ is a nonzero element of $\blV$ then $\dim V_f\leq \ell$ with equality if and only if $f$ is a completely decomposable. The last ingredient to state the minimum weight classification conjecture \cite[Conj. 5.6]{ghorpade_singh} for Schubert codes $\Calm$, which is the first aim of this article, is the notion of Schubert decomposability that we are borrowing from \cite[Def. 2.1]{ghorpade_singh} which is:

\begin{definition}
An element $f \in \bmlV$ is said to be \emph{Schubert decomposable} (with respect to the Schubert variety $\Oalm$) if $f$ is completely decomposable, i.e., $f \neq 0$ and $f=f_1 \wedge \dots \wedge f_{m-\ell}$ for some $f_1, \dots, f_{m-\ell} \in V$, and moreover,
\begin{equation}
    \dim(V_f \cap A_{p_i}) =\alpha_{p_i}-p_i \text{ for all } i=1, \dots, u.
\end{equation}
\end{definition}
Note that in the case when $\alpha$ is a consecutive sequence, the notion of Schubert decomposability is same as the notion of complete decomposability of an exterior product.  In order to give a classification of the minimum weight codewords of Schubert codes $\Calm$, Ghorpade and Singh \cite[Conj. 5.6]{ghorpade_singh} proposed the following conjecture:
\begin{conjecture}
  \label{conj:1} 
  Minimum weight codewords of the Schubert code $C_{\alpha}(\ell,m)$ are precisely the codewords corresponding to Schubert decomposable elements of $\bigwedge^{m-\ell}V$.
\end{conjecture}

Our aim is to settle this conjecture in full generality. The proof is motivated by the ideas used in \cite{ghorpade_singh}, and we use several results from that article in settling the conjecture. For the sake of completeness, we would like to list down all the results and proofs that we will be using in the proof of the main theorem. Following \cite{ghorpade_singh}, for $\ell>1$, we  denote by $\alpha'=(\alpha_1,\ldots,\alpha_{\ell-1})$, the dimension sequence of the truncated partial flag $A_1 \subset \cdots \subset A_{\ell-1}$. For $f \in \bigwedge^{m-\ell}V$, define
$$
E:=\{x \in A_{\ell}: c_{f \wedge x} \text{ is the zero codeword in } C_{\alpha'}(\ell-1,m)\} \quad \text{and} \quad F:= A_\ell \setminus E.
$$
The following lemma is due to Xiang \cite[Lemma 3]{xiang} and it will be used several times throughout this article. 
	 
\begin {lemma} \label{a_lminust_in_e}
    Assume that $\ell > 1$ and $f\in\bigwedge^{m-\ell }V$ is given. Let $E$ be as above and let $t$ be a nonnegative integer such that $\codim_{A_\ell}E \le t$. Then $A_{\ell-t}\subseteq E$.
\end{lemma}	

The next lemma is also due to Xiang \cite[Corollary 1]{xiang} and an alternative coordinate-free proof was given by Ghorpade and Singh in \cite[Lemma 3.3]{ghorpade_singh}.

\begin{lemma} \label{counting_fibres}
    Assume that $\ell > 1$ and $f\in\bigwedge^{m-\ell }V$ is given. Let $\alpha'$ and $E$ be as defined above. Also, let 
    $$
    Z(\alpha',f)=\{(L',x)\in \Omega_{\alpha'}(\ell-1,m) \times A_\ell : (f\wedge x)(L')\neq 0\}
    $$
    and let $\phi:Z(\alpha',f)\longrightarrow W(f)$ be the map given by $(L',x)\mapsto \langle L',x\rangle$. Then
    $\phi$ is well-defined and surjective. Moreover, given any $L \in W(f)$, the following holds.
    \begin{enumerate}
    \item[(i)]	If $L\not\subseteq A_{\ell-1}$ then $|\phi^{-1}(L)|=q^{\ell-1}(q-1)$.
    \item[(ii)]	If $L\subseteq A_{\ell-1}$ and if $\, t:= \codim_{A_\ell}E $, then $|\phi^{-1}(L)|\leq q^{\ell-1} (q^t-1)$.
    \end{enumerate}	
\end{lemma}	
 
The following lemma is proved by Ghorpade and Singh \cite[Lemma 3.4]{ghorpade_singh}. Since the proof of the classification conjecture uses some part of this proof, for the sake of completeness, we borrow the proof as well from \cite{ghorpade_singh}.

\begin{lemma} \label{weight_inequality}
    Assume that $\ell > 1$. Let $f\in\bigwedge^{m-\ell }V$ be such that $c_f\ne 0$ and let $E$ and $F$ be the corresponding sets as defined above. Also, let $t:= \codim_{A_{\ell}}E$. Then 
    \begin{equation} \label{eq_weight_inequality}
    \wt(c_f) \geq 
        \frac{1}{ q^{\ell-1}(q-1)} \sum\limits_{x\in F \setminus A_{\ell-1} }  \wt(c_{f\wedge x}) + 
        \frac{1}{q^{\ell-1}(q^t-1)}\sum\limits_{x\in F\cap A_{\ell-1}} \wt(c_{f\wedge x}).
    \end{equation}
    Moreover, the above inequality is strict if the inequality in part (ii) of Lemma \ref{counting_fibres} is strict for some $L\in W(f)$ with $L\subseteq A_{\ell-1}$. 
\end{lemma}

\begin{proof}
    Let $\alpha'$, $Z(\alpha',f)$ and the map
    $$
    \phi:Z(\alpha',f)\longrightarrow W(f)
    $$
    be as defined in Lemma \ref{counting_fibres}. Define
    $$
    W_1 := \{L \in W(f) : L \nsubseteq A_{\ell-1}\}
    \quad \text{and} \quad
    W_2 := \{L \in W(f) : L \subseteq A_{\ell-1}\}.
    $$
    Let $\theta_1 = |W_1|$ and $\theta_2 = |W_2|$. Then,
    $$
    \wt(c_f)=\theta_1+\theta_2.
    $$
    Partitioning W(f) according to whether its elements are contained in $A_{\ell-1}$, we obtain
    $$
    |Z(\alpha',f)|
    =\sum_{L\in W(f)}|\phi^{-1}(L)|
    =\sum_{L\in W_1}|\phi^{-1}(L)|
    +\sum_{L\in W_2}|\phi^{-1}(L)|.
    $$
    
    On the other hand, counting the fibers of the natural projection
    $$
    Z(\alpha',f)\longrightarrow F,
    $$
    gives
    $$
    |Z(\alpha',f)|
    =\sum_{x\in F}\wt(c_{f\wedge x})
    =\sum_{x\in F\setminus A_{\ell-1}}\wt(c_{f\wedge x})
    +\sum_{x\in F\cap A_{\ell-1}}\wt(c_{f\wedge x}).
    $$

    Now, let $L'\in\Omega_{\alpha'}(\ell-1,m)$ and $x\in A_\ell$. Then
    $$
    L'+\langle x\rangle\in W(f)
    \quad\text{and}\quad
    L'+\langle x\rangle\nsubseteq A_{\ell-1}
    $$
    if and only if
    $$
    x\in F\setminus A_{\ell-1}
    \quad\text{and}\quad
    L'\in W(f\wedge x).
    $$
    Consequently,
    $$
    \sum_{L\in W_1}
    |\phi^{-1}(L)|
    =
    \sum_{x\in F\setminus A_{\ell-1}}
    \wt(c_{f\wedge x}),
    \quad \text{and} \quad
    \sum_{L\in W_2}
    |\phi^{-1}(L)|
    =
    \sum_{x\in F\cap A_{\ell-1}}
    \wt(c_{f\wedge x}).
    $$
    \noindent
    Applying Lemma \ref{counting_fibres} to estimate the fiber sizes gives
    $$
    \sum_{x\in F\setminus A_{\ell-1}}
    \wt(c_{f\wedge x})
    =\theta_1q^{\ell-1}(q-1)
    \quad \text{and} \quad
    \sum_{x\in F\cap A_{\ell-1}}
    \wt(c_{f\wedge x})
    \le
    \theta_2q^{\ell-1}(q^t-1).
    $$
    Combining these identities gives (\ref{eq_weight_inequality}). Moreover, if
    $$
    |\phi^{-1}(L)|<q^{\ell-1}(q^t-1)
    $$
    for some $L\subseteq A_{\ell-1}$, then the second estimate above is strict, and hence so is the inequality in (\ref{eq_weight_inequality}).
\end{proof}
As we stated earlier, several aspects of the Conjecture \cite[Conj. 5.6]{ghorpade_singh} was proved in the article itself. For example, Ghorpade-Singh proved the following theorem.

\begin{theorem} \label{minimum_weight_codewords}
	$d(C_{\alpha}(\ell,m)) = q^{\delta(\alpha)}$. Moreover, if $\ell > 1$ and $f\in \bigwedge^{m-\ell}V$ is such that $c_f$ is a minimum weight codeword in $C_{\alpha}(\ell,m)$, then $c_{f\wedge x}$  is a minimum weight codeword in $C_{\alpha'}(\ell-1,m)$ for every $x\in F$ and furthermore, we must have either (i) $t =1$ and $t'=0$, or (ii) $t' = t \ge 2$,  $\alpha_{\ell} - \alpha_{\ell -1} = 1$, and equality holds in (\ref{eq_weight_inequality}). Here $\alpha', E$ and $F$ are as before, while $t:= \codim_{A_{\ell}}E$  
    and $t':= \codim_{A_{\ell-1}}(E \cap A_{\ell-1})$.
\end{theorem}
This theorem is going to play a key role in settling the classification conjecture. The following remark is useful and hence is worth adding it. 
\begin{remark} \label{theta_2}
    Suppose that $c_f$ is a minimum weight codeword of $C_{\alpha}(\ell,m)$ for some $f \in \bigwedge^{m-\ell}V $. Let $\theta_1:=|\{L \in W(f) : L \nsubseteq A_{\ell-1}\}| \text{ and } \theta_2:=|\{L \in W(f) : L \subseteq A_{\ell-1}\}|$ as in the proof of Lemma \ref{weight_inequality}. By Theorem \ref{minimum_weight_codewords}, $f \wedge x$ is a minimum weight codeword in $C_{\alpha'}(\ell-1,m)$ for all $x \in F$ and hence for $x \in F \cap A_{\ell-1}$. With equality in (\ref{eq_weight_inequality}), we have
    $$
    \theta_2=\dfrac{1}{q^{\ell-1}(q^t-1)}\sum_{x\in F \cap A_{\ell-1}} \wt(c_{f \wedge x}) 
    = \dfrac{1}{q^{\ell-1}(q^t-1)}|F \cap A_{\ell-1}|\,q^{\delta(\alpha')}.
    $$  
\end{remark}


\section{A Classification of Minimum Weight Codewords of $\Calm$}
In this section, we aim to settle the Conjecture \ref{conj:1}. It has already been proved \cite{ghorpade_singh} that if $f\in\bmlV$ is a Schubert decomposable element, then the codeword $c_f\in \Calm$ is a minimum weight codeword. Further, Ghorpade and Singh also proved that if $f \in \bigwedge^{m-\ell} V$ is decomposable such that the codeword $c_f\in \Calm$ is a minimum weight codeword, then $f$ is Schubert decomposable. Therefore, to settle the conjecture and give a complete characterization of the minimum weight codewords of Schubert codes, it is enough to prove that every minimum weight codeword of $C_{\alpha}(\ell,m)$ is of the form $c_h$ for some decomposable $h \in \bigwedge^{m-\ell} V$. As we have remarked, in case when $\alpha$ is completely consecutive then $\Omega_{\alpha}(\ell,m)$ is precisely $G(\ell, A_\ell)$ and the Schubert code $\Calm$ is the Grassmann code $C(\ell, \alpha_\ell)$. Further, in this case, the notion of Schubert decomposability coincides with the notion of complete decomposability. Hence, the conjecture is true by Nogin \cite{nogin_minimum_distance}. Consequently, we assume that $\alpha$ is not completely consecutive. Further, the Schubert variety in the case $\ell=1$ is the projective space $\mathbb{P}(A_1)$, and hence we may assume $\ell\geq 2$. The following Lemma will play a key role in proving the conjecture.

\begin{lemma} \label{a}
 	Let $f \in\bigwedge^{m-\ell}V$ be such that $c_f$ is a minimum weight codeword of $C_\alpha(\ell,m)$ and let $E$ and $F$ be the corresponding sets as defined earlier. If $t:=\codim_{A_{\ell}} E \geq 2$, then  
    $$
    c_{f}=c_{e_{\alpha_\ell} \wedge f_1},
    $$
    where $f_1 \in \bigwedge^{m-\ell-1}V', e_{\alpha_\ell} \in (E \cap A_{\ell}) \setminus A_{\ell-1}$ and $V^\prime = \langle \{e_1, \ldots, e_m\} \setminus \{e_{\alpha_\ell}\} \rangle$.
 \end{lemma}
 
 \begin{proof}
    Since $c_f$ is a minimum weight codeword of $\Calm$ and $t \geq 2$, by Theorem \ref{minimum_weight_codewords}, we have $t'=t \geq 2 \text{ and } \alpha_\ell -\alpha_{\ell -1}=1$. Consequently, $\dim (E \cap A_{\ell-1} )= \dim E-1=\alpha_\ell-t-1$. Thus, we can choose a basis $\{e_1, \ldots, e_m\}$ of $V$ such that
    $$
    A_{\ell-1}=\langle e_1,\ldots,e_{\alpha_{\ell}-1} \rangle, \quad A_{\ell}=\langle e_1,\ldots,e_{\alpha_\ell} \rangle \quad 
    \text{and} \quad
    E=\langle e_1,\ldots,e_{\alpha_\ell-t-1},e_{\alpha_{\ell}}\rangle.
    $$
    Let $V'=\langle \{e_1, \ldots, e_m\} \setminus \{e_{\alpha_\ell}\} \rangle$. Note that we can write $f \in \bigwedge^{m-\ell} V$ as $f=e_{\alpha_\ell} \wedge f_1 +f_2,$ where $f_1 \in \bigwedge^{m-\ell-1} V'$ and $f_2 \in \bigwedge^{m-\ell} V'$. We claim that 
    $$
    c_f(P)= c_{e_{\alpha_\ell} \wedge f_1}(P), \text{ for all }P\in \Oalm.
    $$

    \noindent
    Let $L \in \Omega_\alpha(\ell,m)$ and let $P=u_\ell \wedge \cdots \wedge u_1$ with $u_i \in A_i$ for $ 1 \leq i \leq \ell,$ be a representative of $L$ in $\bigwedge^\ell V$. Two cases arise: either $u_\ell \in A_{\ell-1}$ or $u_\ell \notin A_{\ell-1}$. If $u_\ell \in A_{\ell-1}$ then $L\subset A_{\ell-1}\subset V^\prime$. Consequently, $f_2(L)\in \bigwedge^m V^\prime.$ But $\dim V^\prime =m-1$, therefore we have $f_2(L)=0$. On the other hand, if $u_\ell \notin A_{\ell-1}$ then as $A_{\ell-1}$ is a hyperplane in $A_\ell$, we can write  $u_\ell=v_{\ell-1}+\lambda e_{\alpha_\ell}$ for some $v_{\ell-1} \in A_{\ell-1} \text{ and } \lambda \in \mathbb{F}_q \text{ with } \lambda \neq 0.$ As  in the previous case we get $f_2 \wedge v_{\ell-1} \wedge u_{\ell-1} \wedge \cdots \wedge u_1=0$. Now.
    \begin{align*}
      f_2 \wedge L &=  f_2 \wedge (v_{\ell-1}+\lambda e_{\alpha_{\ell}}) \wedge u_{\ell-1} \wedge \dots \wedge u_1\\
                &= f_2 \wedge v_{\ell-1}  \wedge u_{\ell-1} \wedge \dots \wedge u_1 + f_2 \wedge \lambda e_{\alpha_{\ell}} \wedge u_{\ell-1} \wedge \dots \wedge u_1\\
                &=  f_2 \wedge \lambda e_{\alpha_{\ell}} \wedge u_{\ell-1} \wedge \dots \wedge u_1\\
                &=  (f- e_{\alpha_\ell} \wedge f_1) \wedge \lambda e_{\alpha_{\ell}} \wedge u_{\ell-1} \wedge \dots \wedge u_1\\
                &= \lambda(f \wedge e_{\alpha_\ell} \wedge u_{\ell-1} \wedge \cdots \wedge u_1).
    \end{align*}
    
    Further, since $e_{\alpha_\ell} \in E$,  by definition, $c_{f \wedge e_{\alpha_\ell}}$ is a zero codeword in $C_{\alpha'}(\ell-1,m)$. In particular, $f \wedge e_{\alpha_\ell} \wedge u_{\ell-1} \wedge \cdots \wedge u_1=0$. This implies that $c_{f_2}(P)=0$. Consequently,  
    $$
    c_f(P)=c_{e_{\alpha_\ell} \wedge f_1}(P) \text{ for all }P\in\Oalm.
    $$
\end{proof}

\begin{remark} \label{b}
    Since $\alpha$ is non-consecutive, there exists  an index $i \in \{1, \dots, \ell-1\}$ such that $\alpha_{i+1}-\alpha_i \geq 2$. Let $k$ be the maximum of all such indices, i.e., $k=\max\{i: \alpha_{i+1}-\alpha_i \geq 2\}$. Define
    $$
    \alpha''= \begin{cases}
    \alpha_i & \text{if } i \leq k, \\
    \alpha_i-1 & \text{if } i > k.
    \end{cases}
    $$
    Since $k$ is the largest index for which
    $\alpha_{j+1}-\alpha_j\ge2$, the entries
    $\alpha_{k+1},\ldots,\alpha_\ell$ are consecutive. Consequently,
    $\alpha''\in I(\ell,m-1)$. Note that under the hypothesis of the above lemma, $f_1$ naturally corresponds to a codeword in $C_{\alpha''}(\ell,m-1)$.  
\end{remark}

 \begin{lemma} \label{c}
    Let $f \in\bigwedge^{m-\ell}V$ be such that $c_f$ is a minimum weight codeword of $C_\alpha(\ell,m)$, and let  $t:=\codim_{A_{\ell}} E \geq 2$. If $f_1\in\bigwedge^{m-1-\ell}V^\prime$ be the exterior product obtained in the Lemma \ref{a} then  the codeword $c_{f_1}\in C_{\alpha''}(\ell,m-1)$ is a minimum weight codeword.
 \end{lemma}
 \begin{proof}
    First, we claim that $c_{f_1}\in C_{\alpha''}(\ell,m-1) $ is a nonzero codeword. Suppose on the contrary $c_{f_1}=0$. For any $L\in\Oalm$, either $L\subseteq A_{\ell-1}$ or $L\nsubseteq A_{\ell-1}$. If $L\subseteq A_{\ell-1}$ then $L\in  \Omega_{\alpha''}(\ell,m-1)$. Consequently, $c_{f_1}$ being a zero codeword of $C_{\alpha''}(\ell,m-1)$ gives $f(L)=e_{\alpha_\ell} \wedge f_1(L)= 0. $ On the other hand, if $L\nsubseteq A_{\ell-1}$, we can find a representative $P=u_\ell \wedge \cdots \wedge u_1$ of $L$ satisfying with $u_i \in A_i$ for $ 1 \leq i \leq \ell,$ and $u_\ell\in A_\ell\setminus A_{\ell-1}.$ Now, since $t\geq 2$ and $f$ is a minimum weight codeword of $\Calm$, we have $A_{\ell-1}$ is a hyperplane in $A_\ell$ and hence, we can write $u_\ell=v_{\ell-1}+\lambda e_{\alpha_\ell}$ for some $v_{\ell-1} \in A_{\ell-1} \text{ and nonzero } \lambda \in \Fq$. Now, 
    \begin{align*}
        f(L)&= {e_{\alpha_\ell} \wedge f_1}(L)\\
        &=e_{\alpha_\ell} \wedge f_1 \wedge u_\ell \wedge \cdots \wedge u_1\\
        &=e_{\alpha_\ell} \wedge f_1 \wedge v_{\ell-1} \wedge \cdots \wedge u_1\\
        &=0
    \end{align*}
    where the last equality follows as, either $v_{\ell-1} \wedge \cdots \wedge u_1=0$ or $v_{\ell-1} \wedge \cdots \wedge u_1$ represents a point in $\Omega_{\alpha''}(\ell,m-1)$. This implies, in particular, that $f(L)=0$ for representatives of all the points of $\Oalm$ and hence must be a zero codeword. But this contradicts the fact that $c_f$ is  a minimum weight codeword. Thus, we have $c_{f_1}\neq 0.$

    Next, let $W_1$, $W_2$, $\theta_1$, and $\theta_2$ be as in Lemma \ref{weight_inequality}. We know that $\wt(c_f)=\theta_1+\theta_2.$ Let $W_{\alpha''}(f_1)$ be the support of $f_1$ in $\Omega_{\alpha''}(\ell,m-1)$, namely,
    $$
    W_{\alpha''}(f_1)
    :=
    \{L \in \Omega_{\alpha''}(\ell,m-1) : f_1(L)\neq 0\}.
    $$
    Observe that if $L\in W_{\alpha''}(f_1)$ then $L\subset A_{\ell-1}$ and $f_1(L)\neq 0$. Further, since $e_{\alpha_\ell}\notin V^\prime$ and $f_1(L)\in \bigwedge^{m-1}V^\prime$, $e_{\alpha_\ell}\wedge f_1(L)\neq 0$. Consequently $L\in W_2$. In other words,  $W_{\alpha''}(f_1)\subseteq W_2$. Consequently, $ \left|W_{\alpha''}(f_1)\right|\le \theta_2 $. As $c_{f_1}$ is a non-zero codeword in $C_{\alpha''}(\ell,m-1)$, we have 
    $$
    q^{\delta(\alpha'')}=q^{\delta(\alpha)-(\ell-k)} \leq \left|W_{\alpha''}(f_1)\right|\le \theta_2.
    $$
    Further, since $c_f$ is a minimum weight codeword in $C_{\alpha}(\ell,m)$, by Remark \ref{theta_2}, we have 
    \begin{align*}
      \theta_2 &= \dfrac{1}{q^{\ell-1}(q^t-1)}|F \cap A_{\ell-1}| \, q^{\delta(\alpha')} \\
               & =\dfrac{1}{q^{\ell-1}(q^t-1)}(q^{\alpha_{\ell}-1}-q^{\alpha_{\ell}-1-t})\,q^{\delta(\alpha)-\alpha_\ell+\ell} \\
               &=q^{\delta(\alpha)-t}.  
    \end{align*}

    Note that $\wt(c_{f_1})\leq \theta_2$ implies that $t\leq \ell-k$ and proving the equality guarantees that $c_{f_1}$ is a minimum weight codeword. Suppose, to the contrary, that $t < \ell-k$, that is, $t=\ell-k-j$ for some positive integer $j$. Then by Lemma \ref{a_lminust_in_e}, we have
    $$
    A_{k+j} = A_{\ell-t} \subset E
    $$
  Since $\alpha_{k+1},\ldots,\alpha_{k+j+t}$ are consecutive with $\alpha_{k+j+t}=\alpha_\ell$, we get
    $$
    \dim A_{k+j}
    =\alpha_\ell-t
    =\dim E
    $$
   Therefore $A_{\ell-t}=E$. Further, as $t\ge2$, it follows that
    $$
    A_{\ell-t}\subset A_{\ell-1}.
    $$
    Consequently, $E=A_{\ell-t} \subset A_{\ell-1}$ and
    $$
    t'
    =\codim_{A_{\ell-1}}(E\cap A_{\ell-1})
    =\alpha_{\ell-1}-1-\dim(E\cap A_{\ell-1})
    =t-1 < t,
    $$
    which is a contradiction. Therefore, $t=\ell-k$, and hence
    $c_{f_1}$ is a minimum weight codeword in
    $C_{\alpha''}(\ell,m-1)$.
    \end{proof}

\begin{theorem} \label{maintheorem}
    Assume that $\alpha$ is non-consecutive and $\Calm$ be the corresponding Schubert code. If $c$ is  a minimum weight codeword of $C_{\alpha}(\ell,m)$, then $c = c_h$ for some decomposable $h\in\bigwedge^{m-\ell}V$.
\end{theorem}
 
\begin{proof}
    We will use induction on $m+\ell$. The result clearly holds for arbitrary $m$ when $\ell=1$ since every nonzero element of $\bigwedge^{m-1}V$ is decomposable. In \cite{ghorpade_singh}, Ghorpade and Singh proved that the result holds affirmatively for arbitrary $m$ when $\ell=2$.
    Now, suppose   $m+\ell > 2$ and the result is true for all values of $m+\ell$ smaller than the given one. Let $c$ be a minimum weight codeword of $C_{\alpha}(\ell,m)$. Fix $f\in \bigwedge^{m-\ell}V$ such that $c=c_f$, and let $E, F$ be same as defined earlier.
    We shall now divide the proof into two cases. 

    \medskip
    \noindent
    \textbf{Case 1.}  $t=1$.
    This case has already been dealt by Ghorpade and Singh \cite[Thm. 6.1]{ghorpade_singh}. But for the sake of completeness of the proof, we borrow their proof. 
    Since $\codim_{A_\ell}E = 1$, by Lemma \ref{a_lminust_in_e}, $A_{\ell -1} \subseteq E$. Thus, we can choose a basis $\{e_1, \dots , e_m\}$ of $V$ such that 
    $$
    A_i = \langle e_1, \dots , e_{\alpha_i}\rangle \; \text{ for } \; i=1, \dots, \ell \quad \text{and} \quad E = 
    \langle e_1, \dots , e_{\alpha_\ell - 1}\rangle.
    $$
    Let $x:= e_{\alpha_\ell}$. Clearly, $x\in F$ and hence $f\wedge x$ corresponds to a minimum weight codeword of $C_{\alpha'}(\ell-1,m)$. By induction, $c_{f\wedge x} = c_g$ for some decomposable $g\in \bigwedge^{m-\ell + 1}V$. Moreover, since $g$ is decomposable and $c_g$ is a minimum weight codeword in $C_{\alpha'}(\ell-1,m)$, we have that $g$ is Schubert decomposable and therefore $\dim V_g \cap A_{\ell-1} = \alpha_{\ell-1} -(\ell-1)$. Thus, we can recursively choose 
    $z_1, \dots , z_{\ell-1}$ such that 
    $$
    z_i \in A_{\ell-1} \setminus \left( \langle z_1, \dots , z_{i-1} \rangle +  V_g \cap A_{\ell-1} \right). 
    $$
    In particular, $z_1, \dots , z_{\ell - 1}$ span an $(\ell-1)$-dimensional subspace, say $B_{\ell-1}$ of $A_{\ell-1}$  such that $A_{\ell-1}= B_{\ell-1} + \left( V_g \cap A_{\ell-1} \right)$. This implies that $V_g \cap B_{\ell-1} = \{0\}$. Also since $\dim V_g = m -\ell +1$ and $x\not\in A_{\ell -1}$, we see that $\dim V_g \cap ( B_{\ell-1} + \langle x \rangle ) \ge 1$. Hence, $V_g$ contains an element of the form $b+x$ for some $b \in B_{\ell -1}$. Consequently, we can find 
    $g_1, \dots, g_{m-\ell}\in V$ such that $g_1, \dots, g_{\alpha_{\ell-1} -(\ell -1)}$ spans  $V_g \cap A_{\ell-1}$, and 
    $$
    g = g_1 \wedge \dots \wedge g_{m-\ell} \wedge (b+x)  = g'\wedge b + g'\wedge x, \quad \text{where } \; 
    g':= g_1 \wedge \dots \wedge g_{m-\ell}.
    $$
    Note that $V_g \cap A_{\ell-1} = \langle g_1, \dots, g_{\alpha_{\ell-1} -(\ell -1)} \rangle \subseteq V_{g'} \cap A_{\ell-1}
    \subseteq V_g \cap A_{\ell-1}$. Thus,  
    $$
    V_{g'} \cap A_{\ell-1} = V_g \cap A_{\ell-1} \quad \text{and} \quad \dim ( V_{g'} \cap A_{\ell-1} ) = \alpha_{\ell-1} - (\ell-1). 
    $$
    \noindent
    \textbf{Claim:} $c_{g'\wedge b } = 0$.\newline
    \noindent 
    The claim clearly is true if $b\in V_{g'}$. Now if $b\not\in V_{g'}$ then 
   \begin{align*}
       V_{g' \wedge b} \cap A_{\ell-1} 
    &=  ( V_{g'} + \langle b \rangle) \cap A_{\ell-1} \\
    &=  ( V_{g'} \cap A_{\ell-1} ) + \langle b \rangle.
   \end{align*}
    Note that as $\dim  (V_{g' \wedge b} \cap A_{\ell-1}) = \alpha_{\ell-1} - (\ell -1)+1 $ therefore $V_{g' \wedge b}$ has a nontrivial intersection with any $(\ell-1)$-dimensional subspace of $A_{\ell-1}$. In particular, we get that $ U \cap V_{g' \wedge b}  \ne \{0\}$ for every $U\in \Omega_{\alpha'}(\ell-1,m)$. This proves the claim.
        
    The claim gives that $c_g = c_{g'\wedge x}$. Now, writing each of $g_1, \dots, g_{m-\ell}$ as a linear combination of $e_1, \dots , e_m$ and noting that $x= e_{\alpha_{\ell}}$, we get
    $g'\wedge x = h \wedge x$, where $h $ is a decomposable element of $\bigwedge^{m-\ell}V$ of the form 
    $h_1 \wedge \dots \wedge h_{m-\ell}$, where each of $h_1, \dots, h_{m-\ell}$ lies in the $(m-1)$-dimensional space $V$ spanned by $\{e_1, \dots , e_m\}\setminus\{x\}$. 
   If we can prove that $c_f = c_h$ then \cite[Thm. 5.5]{ghorpade_singh} guarantees that $h$ is Schubert decomposable and we are done. Thus, we now aim to show that  $c_f = c_h$.

    Let $L \in \Omega_{\alpha}(\ell,m)$ and let $P= u_{\ell} \wedge \dots \wedge u_1$ with $u_i \in A_i$ for $1\le i \le \ell$, be a representative of $L$ in $\bigwedge^{\ell}V$. We want to show that 
    $$
    c_f(P)= c_h(P).
    $$
     Since $c_{f\wedge x} = c_{h \wedge x}$ , we readily see that $f \wedge x \wedge u_{\ell -1}\wedge \dots \wedge u_{1} = h \wedge x \wedge u_{\ell -1}\wedge \dots \wedge u_{1}$. We will now consider two cases. First, suppose $u_{\ell} \in E$. Then $c_{f \wedge u_{\ell} }$ is a zero codeword in $C_{\alpha'}(\ell-1,m)$ and hence $c_f(P)=0$. On the other hand, by our choice of $h$, we see that  $V_h + L$ is a subspace of $V'$. Since $\dim V_h + \dim L = m > \dim V'$, we must have $V_h \cap L \ne \{0\}$ and so we obtain $c_h(L) =0$ as well. Now suppose $u_{\ell} \not\in E$. Then $u_{\ell} = v_{\ell} + \lambda x$ for a unique $v_{\ell} \in E$ and $\lambda \in \Fq$ with $\lambda \ne 0$. As in the previous case,
    $f \wedge v_{\ell} \wedge u_{\ell-1} \wedge \dots \wedge u_{1}  = 0 = h \wedge v_{\ell} \wedge u_{\ell-1} \wedge \dots \wedge u_{1} $. Consequently, 
    $$
    c_h(P) = \lambda \left( h \wedge x \wedge u_{\ell -1}\wedge \dots \wedge u_{1}  \right) 
    = \lambda (f \wedge x \wedge u_{\ell -1}\wedge \dots \wedge u_{1}) 
    = f \wedge u_{\ell} \wedge u_{\ell-1} \wedge \dots \wedge u_{1},
    $$
    and thus $ c_h(P) =  c_{f} (P)$. 

    \medskip
    \noindent
    \textbf{Case 2.}  $t \geq 2$. In this case, Theorem \ref{minimum_weight_codewords} implies that $t = t' \geq2$ and $\alpha_{\ell}-\alpha_{\ell-1}=1$. Consequently, $\dim (E \cap A_{\ell-1})= \dim E-1=\alpha_\ell-t-1$. Thus, we may choose a basis $\{e_1, \ldots, e_m\}$ of $V$ such that
    $$
    A_{\ell-1}=\langle e_1,\ldots,e_{\alpha_{\ell}-1} \rangle, \quad A_{\ell}=\langle e_1,\ldots,e_{\alpha_\ell} \rangle 
    \quad \text{and} \quad
    E=\langle e_1,\ldots,e_{\alpha_\ell-t-1},e_{\alpha_{\ell}}\rangle.
    $$
    By Lemma \ref{a}, we have that $c_{f}=c_{e_{\alpha_\ell} \wedge f_1}$ where $f_1 \in \bigwedge^{m-\ell-1}V'$. Moreover,  Lemma \ref{c} guarantees that $c_{f_1}$ is a minimum weight codeword in $C_{\alpha''}(\ell,m-1)$. By induction there exist a decomposable $h' \in \bigwedge^{m-1-\ell}V'$  such that $c_{f_1}=c_{h'}$. Next, we want to show that the codeword $c_f$ can be indexed by $e_{\alpha_\ell} \wedge h'$ as well, i.e., $c_f=c_{e_{\alpha_\ell} \wedge h'}$.

    It is enough to show that $c_{e_{\alpha_\ell} \wedge f_1}=c_{e_{\alpha_\ell} \wedge h'}$. Let $L \in \Omega_{\alpha}(\ell,m)$ and let $P= u_{\ell} \wedge \dots \wedge u_1$  be a  representative of $L$ in $\bigwedge^{\ell}V$ satisfying $u_i \in A_i$ for $1\le i \le \ell$. If $u_\ell \in A_{\ell-1}$, then $L \in \Omega_{\alpha''}(\ell,m-1)$ and $f_1\wedge u_\ell \wedge \cdots \wedge u_1=h\wedge u_\ell \wedge \cdots \wedge u_1 $.  Hence $c_{e_{\alpha_\ell} \wedge f_1}(P)=c_{e_{\alpha_\ell} \wedge h'}(P)$. On the other hand, if $u_\ell \notin A_{\ell-1}$ then $u_\ell=v_{\ell-1}+\lambda e_{\alpha_\ell}$  for some $v_{\ell-1} \in A_{\ell-1} \text{ and  nonzero } \lambda \in \Fq.$ Now
    $$
    c_{e_{\alpha_\ell} \wedge f_1}(P)=e_{\alpha_\ell} \wedge f_1 \wedge u_\ell \wedge \cdots \wedge u_1=e_{\alpha_\ell} \wedge f_1 \wedge v_{\ell-1} \wedge \cdots \wedge u_1 =e_{\alpha_\ell} \wedge h' \wedge v_{\ell-1} \wedge \cdots \wedge u_1.
    $$
    Set $h=e_{\alpha_\ell} \wedge h'$, then, we get that $c_f=c_h$ where $h \in \bigwedge^{m-\ell}$ is decomposable. This completes the proof.
    \end{proof}	
Combining the above theorem with \cite[Thm. 5.5]{ghorpade_singh}, we establish that the Conjecture \ref{conj:1} holds affirmatively for all $\alpha \in I(\ell,m)$. More precisely, for any $c \in \Calm$
$$
c \text{ has minimum weight} \iff c = c_h \text{ for some Schubert decomposable } h \in \bigwedge^{m-\ell}V.
$$
Ghorpade and Singh in \cite[Thm. 7.5]{ghorpade_singh} enumerate the number of Schubert decomposable elements of $\bigwedge^{m-\ell}V$. In view of the above characterization, this gives us the number of minimum weight codewords of $\Calm$ for all $\alpha \in I(\ell,m)$.

    
\section{Characterization of second minimum weight codewords of Grassmann code}

Computing the weight spectrum of the Grassmann code is a challenging problem in general. The first step in the direction of finding the weight enumerator polynomial of the Grassmann code is to determine the second minimum weight of these codes and enumerate them. Recently, Datta and Dutta \cite{datta2026secondminimumweightgrassmann} computed the  second minimum weight of the Grassmann code $C(\ell,m)$ and showed that for  $2 \leq \ell \leq m-2$, the second minium weight is 
 $$
 q^{\delta-2}(q^2+1),
 $$
where $\delta=\ell(m-\ell)$. Thus, it is natural to find a classification of all such codewords and do their enumeration. In this section, we aim to solve these problems for the Grassmann code $C(\ell, m)$ for any $\ell, $ and any $m$. Note that we may always assume that $\ell\leq m-\ell$. Also, the case $\ell=1$ corresponds to projective Reed-Muller code of order one, i.e., the constant weight code, we should not be worried about these codes. Therefore, for this section we assume that $\ell$ satisfies $2\leq \ell\leq m-\ell$. The proof of the classification of second minimum weight codewords of the Grassmann code $C(\ell, m)$ uses the classification of the minimum weight codewords of the Schubert code that we proved in Theorem \ref{maintheorem}. We begin by fixing some notations and conventions. Let $\theta = (m-\ell-1, m-\ell+2, \dots, m)\in I(\ell, m)$ be  fixed and let 
$$
A_1 \subset \dots \subset A_\ell
$$
be a partial flag corresponding to the dimension sequence $\theta$. Let $C_\theta(\ell, m)$ be the Schubert code corresponding to the Schubert variety $\Otlm$ and let $k_\theta=\dim \Ctlm$. It follows from \cite{Ghorpade_and_Tsfasman} that 
$$
k_\theta =|\{\beta\in I(\ell, m): \beta\leq \theta\}|,
$$
where $\leq$ is the Bruhat order defined earlier. Recall that, if $|\Otlm|=n_\theta$ then the Schubert code $\Ctlm$ is the image of the evaluation map  
\begin{equation}
    \label{eq:EvSchubert}
    \Ev : \bigwedge^{m-\ell}V\longrightarrow \mathbb{F}_{q}^{\,n_{\theta}} \quad \text{defined by} \quad f \mapsto c_f,
\end{equation}
where $c_f$ is the codeword obtained by evaluation $f$ on a fixed, ordered representatives of points of $\Otlm$. Note that the elements $\beta\in I(\ell,m)$ satisfying $\beta \nleq \theta$ are precisely the $\ell+1$ tuples obtained by removing the $i$-th coordinate from the $(\ell+1)$ tuple $(m-\ell, \dots , m)$. Therefore, we have $|I(\ell,m)| - k_\theta = \ell+1$, i.e., $\binom{m}{\ell}-k_\theta=\ell+1$, which is precisely the dimension of the kernel of the evaluation map in equation \eqref{eq:EvSchubert}. Further, let $\{v_1, \dots, v_{m-\ell-1}\}$ be a basis of $A_1$. Extend this to a basis $\{v_1, \dots, v_m\}$ of $V$. Observe that for any $v \in V$, the element $v_1 \wedge \dots \wedge v_{m-\ell-1} \wedge v$ belongs to $\ker(\Ev)$. Moreover, the elements $v_1 \wedge \dots \wedge v_{m-\ell-1} \wedge v_{m-\ell+i}$ for $i \in \{0, \dots, \ell\}$ are linearly independent. Since $\dim(\ker(\Ev))=\ell+1$, they form a basis of $\ker(\Ev)$. Thus, any element of $\ker(\Ev)$ will be of the form $v_1 \wedge \dots \wedge v_{m-\ell-1} \wedge v$ for some $v \in V$. We will use this observation in proving the following theorem.

\begin{theorem} \label{2a}
    Let $f \in \bigwedge^{m-\ell}V$ be such that the corresponding codeword $c_f\in C(\ell, m)$ is a codeword with $\wt(c_f) = q^{\ell(m-\ell)-2}(q^2+1)$. Then there exist linearly independent vectors $v_1, \dots, v_{m-\ell-2} \in V$ and an element $\xi \in \bigwedge^2 V$ of rank 4 such that 
     $$
        f=v_1 \wedge \dots \wedge v_{m-\ell-2} \wedge \xi.
     $$   
\end{theorem}
\begin{proof}
   The proof is done by induction on $m+\ell$. The case $m=4 \text{ and } \ell=2$ follows from Nogin \cite[Thm. 3]{nogin_minimum_distance}. Assume now that $m+\ell > 4$ and the theorem holds for all values of $m+\ell$ smaller than the given one.

    \medskip
    \noindent
    \textbf{Case 1:} If there exist a hyperplane  $W$ of $V$ such that $G(\ell,W) \subseteq V(f)$.
    
    Choose a basis $\{e_1, \dots, e_{m-1}\}$ of $W$ and extend it to a basis $\{e_1, \dots, e_m\}$ of $V$. We can write
    $$
    f=e_m \wedge f_1+f_2,
    $$
    where $f_1 \in \bigwedge^{m-\ell-1} W$ and $f_2 \in \bigwedge^{m-\ell} W$. We claim that $f_1=0$. Note that $f_2(L)= 0$ for all $L\in G(\ell, W)$ as $f_2(L)\in \bigwedge^{m} W$ and $\dim W= m-1$. Hence,
    $$
    f(L)= e_m\wedge f_1 (L) \text{ for all } L\in G(\ell, W).
    $$
    But as $G(\ell,W) \subseteq V(f)$, we have $f(L)=0$ for all $L\in G(\ell, W)$. Consequently, $e_m\wedge f_1 (L)=0$ for all $L\in G(\ell, W)$. Now, note that $e_m\in V \setminus W$ and $f_1 \in \bigwedge^{m-\ell-1}W$, therefore the map  $\varphi : \bigwedge^{m-1}W \mapsto \bigwedge^{m}V $ defined as $\varphi(\omega)=e_m \wedge \omega$ is injective. As a consequence, we get 
    $$
     f_1(L)= 0 \text{ for all } L\in G(\ell, W).
    $$
    But this implies that $f_1=0$. This proves the claim.
    
    Thus, we have $f=f_2 \in  \bigwedge^{(m-1)-(\ell-1)} W$. As observed in \cite[Theorem 4.10]{datta2026secondminimumweightgrassmann} the codeword $c_f$ should be a second minimum weight codeword in $C(\ell-1,m-1)$. By induction,  $f_2$ is of the form $f_2=v_1 \wedge \cdots \wedge v_{m-\ell-2} \wedge \xi$, for some linearly independent vectors $v_1, \ldots, v_{m-\ell-2}$ in $W$ and $\rank{\xi}=4$. Since $f = f_2$ and $W \subset V$ we get  
    $$
    f = v_1 \wedge \cdots \wedge v_{m-\ell-2} \wedge \xi.
    $$

    \medskip
    \noindent
    \textbf{Case 2:} If $G(\ell,W) \nsubseteq V(f)$ for all hyperplane $W$ of $V$.

    By Theorem 4.10 of \cite{datta2026secondminimumweightgrassmann}, there exists a hyperplane $W$ of $V$ such that the restriction of $c_f$ to $G(\ell,W)$, regarded as a codeword of $C(\ell,m-1)$, has minimum weight $q^{\ell(m-1-\ell)}$. \\   
    Choose a basis $\{e_1, \dots, e_{m-1}\}$ of $W$ and extend it to a basis $\{e_1, \dots, e_m\}$ of $V$. As before, we can write 
    $$
    f=e_m \wedge f_1+f_2,
    $$
    where $f_1 \in \bigwedge^{m-\ell-1} W$ and $f_2 \in \bigwedge^{m-\ell} W$. As in the previous case, we have 
    $$
    f_2(L)=0 \text{ for all }L \in G(\ell,W)
    $$
    and hence  $f(L)=e_m \wedge f_1 (L)$ for all $L \in G(\ell,W)$. Injectivity of the map $\varphi$ defined in previous case gives
    
    $$
    f(L) = e_m \wedge f_1(L) = 0 \iff f_1(L)=0.
    $$
   Consequently, $\wt (c_f) = \wt (c_{f_1})$ on $G(\ell,W)$ and hence $c_{f_1}$ is a minimum weight codeword in $C(\ell,m-1)$. But the minimum weight codewords of $C(\ell,m-1)$ are given by decomposable elements \cite{nogin_minimum_distance}, so there exist $w_1, \dots, w_{m-\ell-1}\in W$ such that  $f_1 = w_1 \wedge \dots \wedge w_{m-\ell-1}$. Now consider the partial flag 
    $ A_1 \subset \dots \subset A_\ell $ where 
    $$
     A_1 = \langle w_1, \ldots, w_{m-\ell-1} \rangle, \; A_2 = \langle w_1, \ldots, w_{m-\ell+2} \rangle, \; \dots, \; A_\ell = \langle w_1, \ldots, w_{m} \rangle.
    $$
    Note that $\theta=(m-\ell-1, m-\ell+2, \dots, m)$ is the dimension sequence of the above partial flag. By \cite[Proposition 4.6]{datta2026secondminimumweightgrassmann}, $c_f$ is a minimum weight codeword in $\Omega_{\theta}(\ell,m)$. Hence, by Theorem \ref{maintheorem}, there exists a Schubert decomposable element $g \in \bigwedge^{m-\ell}V$ such that $c_f=c_g$. Since $g$ is Schubert decomposable, $\dim (V_g \cap A_1)=m-\ell-2$. Choose a basis $\{v_1, \dots, v_{m-\ell-2}\}$ for $V_g \cap A_1$ and extend it to a basis of $V$ such that 
    $$
    A_1 =\langle v_1, \ldots, v_{m-\ell-1} \rangle, \; A_2 = \langle v_1, \ldots, v_{m-\ell+2} \rangle, \; \dots, \; A_\ell = \langle v_1, \ldots, v_{m} \rangle.
    $$
    Clearly, $g=v_1 \wedge \cdots \wedge v_{m-\ell-2} \wedge \eta $  where $\eta \in \bigwedge^2 V$. Also note that $c_f=c_g$ in $C_{\theta}(\ell,m)$ implies that the codeword $c_{f-g}$ is the zero codeword in $C_{\theta}(\ell,m)$. Hence, $f-g=h$ where $h= v_1 \wedge \dots v_{m-\ell-1} \wedge v \in \ker(\Ev)$ and $v \in V$. Therefore,
    $$
    \begin{aligned}
    f
    & = v_1\wedge\cdots\wedge v_{m-\ell-2}\wedge\eta
    +v_1\wedge\cdots\wedge v_{m-\ell-1}
    \wedge v \\
    & = v_1\wedge\cdots\wedge v_{m-\ell-2}
    \wedge
    \left(\eta+ v_{m-\ell-1} \wedge v \right).
    \end{aligned}
    $$
    Let $ \xi = \eta + v_{m-\ell-1} \wedge v$. We claim that $\xi$ is of rank 4.
    Clearly, $\xi \neq 0$. Moreover, $\xi$ is not decomposable otherwise $f$ would itself be decomposable and in that case $c_f$ is a minimum weight codeword, i.e.,  $\wt(c_f)=q^{\ell(m-\ell)} < q^{\ell(m-\ell)-2}(q^2+1)$. Since the rank of an alternating form is always even and $0 \leq \rank( \xi) \leq 4$, we get that $\rank(\xi)=4$. This gives the desired form of $f$.

\end{proof}

We next prove the converse, thereby completing the characterization of the second-minimum-weight codewords of $C(\ell,m)$. Before that, we prove the following Lemma that will be used in the next result.

\begin{lemma} \label{d}
   Let $f \in \bigwedge^{m-\ell}$ be of the form
   $$
   f = v_1 \wedge \dots \wedge v_{m-\ell-2} \wedge (v_{m-\ell-1} \wedge v_{m-\ell} + v_{m-\ell+1} \wedge v_{m-\ell+2}),
   $$
   where $v_1, \dots, v_{m-\ell+2}$ are linearly independent vectors in $V$. Then  for $x \in V$,
   $$
    f \wedge x \text{ is decomposable } \iff x\in (V_f + \langle v_{m-\ell-1}, \dots, v_{m-\ell+2} \rangle) \setminus V_f.
   $$
\end{lemma}

\begin{proof}
    First, we prove that $V_f=\langle v_1, \dots, v_{m-\ell-2}\rangle$. The inclusion $\langle v_1, \dots, v_{m-\ell-2}\rangle\subseteq V_f$ is obvious. For the converse inclusion, let $x \in V_f$. Extend $v_1, \dots, v_{m-\ell+2}$ to a basis $\{v_1, \dots, v_{m}\}$ of $V$ and write $x =\sum_{i=1}^{m} c_iv_i$. As $x \in V_f$, we have $f \wedge x=0$. Simplifying this, we get 
    $$
    v_1 \wedge \dots \wedge v_{m-\ell-2} \wedge (v_{m-\ell-1} \wedge v_{m-\ell} + v_{m-\ell+1} \wedge v_{m-\ell+2}) \wedge \sum_{i=m-\ell-1}^{m} c_iv_i=0.
    $$
    Expanding the wedge product and using the linear independence of basis vectors, we get $c_i=0$ for all $i \in \{m-\ell-1, \dots, m\}$. In other words, $x \in \langle v_1, \dots, v_{m-\ell-2}\rangle$  and hence $V_f=\langle v_1, \dots, v_{m-\ell-2}\rangle$.

    It is easy to see that if $x \in (V_f+\langle v_{m-\ell-1}, \dots, v_{m-\ell+2}\rangle) \setminus V_f$, then $f \wedge x$ is completely decomposable. We already saw that  $f \wedge x =0$ if and only if $x \in V_f$. We will show that $f \wedge x$ is non-decomposable for $x\in V\setminus (V_f+\langle v_{m-\ell-1}, \dots, v_{m-\ell+2}\rangle)$. For that assume $x \in V$ such that $x \notin (V_f + \text{span}\{v_{m-\ell-1}, \dots, v_{m-\ell+2} \})$. Then $v_1, \dots, v_{m-\ell+2},x$ are linearly independent. Set $v_{m-\ell+3}=x$ and extend these vectors to a basis $\{v_1, \dots, v_{m-\ell+2}, v_{m-\ell+3}, \dots, v_m\}$ of $V$. Define $\omega_1= v_1 \wedge \dots \wedge v_{m-\ell-2} \wedge v_{m-\ell-1} \wedge v_{m-\ell} \wedge x$ and $ \omega_2= v_1 \wedge \dots \wedge v_{m-\ell-2} \wedge v_{m-\ell+1} \wedge v_{m-\ell+2} \wedge x$. Then $f \wedge x = \omega_1+ \omega_2$. We claim that $\dim(V_{\omega_1} \cap V_{\omega_2}) \leq m-\ell-1$.  Let $y \in V_{\omega_1} \cap V_{\omega_2}$. Then $y$ can be written as
    $$
    y=\sum_{i=1}^{m-\ell} \alpha_i v_i +\alpha_{m-\ell+3}v_{m-\ell+3} = \sum_{i=1}^{m-\ell-2} \beta_i v_i + \sum_{i=m-\ell+1}^{m-\ell+3} \beta_i v_i.
    $$
    By linear independence of the basis vectors and comparison of the above two
    expressions, we get that
    $$
    \alpha_i=\beta_i \quad \text{for } i \in \{1, \dots, m-\ell-2, m-\ell+3\} \quad \text{ and} \quad \alpha_{m-\ell-1}=\alpha_{m-\ell}=0.
    $$
     Thus, $y \in \langle v_1, \dots, v_{m-\ell-2},v_{m-\ell+3}\rangle$ and $\dim(V_{\omega_1} \cap V_{\omega_2}) \leq m-\ell-1$. By \cite[Lemma 3]{ghorpade_decomposable}, we conclude that $ \omega_1 +\omega_2 = f \wedge x$ is not decomposable.
\end{proof}

\begin{theorem} \label{2b}
    Let $f \in \bigwedge^{m-\ell}$ be such that $f=v_1 \wedge \dots \wedge v_{m-\ell-2} \wedge \xi$ where $v_1, \dots, v_{m-\ell-2}$ are linearly independent vectors in $V$ and $\xi \in \bigwedge^2V$ is an alternating form of rank 4. Then 
    $$
    \wt(c_f)= q^{\delta-2}(q^2+1).
    $$
    In other words, $c_f$ is a second minimum weight codeword of $C(\ell,m)$.
\end{theorem}

\begin{proof}
    Clearly, there exists linearly independent vectors $ \{v_1, \dots, v_{m-\ell+2}\}$ such that 
    $$
    f = v_1 \wedge \dots \wedge v_{m-\ell-2} \wedge (v_{m-\ell-1} \wedge v_{m-\ell} + v_{m-\ell+1} \wedge v_{m-\ell+2}).
    $$
    We prove the theorem by induction on $m+\ell$. The theorem clearly holds because of Nogin \cite{nogin_minimum_distance}, for any  $m=4$ and $\ell=2$. Assume now that $m+\ell > 4$ and the result holds for
    all values of $m+\ell$ smaller than the given one.
    
    \noindent
    Consider the set
    $$
    Z:=\{(L',x)\in G_{\ell-1,m}\times V : (f\wedge x)(L')\neq 0\}.
    $$
    Define a map
    $$
    \phi:Z\longrightarrow W(f),\qquad (L',x)\longmapsto \langle L',x\rangle.
    $$
    By the definition of $Z$,  $\phi$ is
    a well-defined,  surjective map. To compute the weight of $c_f$, we will count the cardinality of $Z$ in two different ways. 
     For a fixed $x\in V$, the number of
    $L'\in G_{\ell-1,m}$ for which $(f\wedge x)(L')\neq0$ is $\wt(c_{f\wedge x})$. Thus
    $$
    |Z|=\sum_{x \in V}\wt(c_{f \wedge x}).
    $$
    If $x \in V_f$, then $f \wedge x=0$, so these elements do not contribute. By Lemma \ref{d}, the form $f \wedge x$ is decomposable if and only if $x\in (V_f + \langle v_{m-\ell-1}, \dots, v_{m-\ell+2} \rangle) \setminus V_f$ and the number of such $x$ is
    $q^{m-\ell+2}-q^{m-\ell-2}$. For the remaining $q^m-q^{m-\ell+2}$ choices of $x$, $f \wedge x$ is non-decomposable. Note that $f \wedge x$ is a codeword in $C(\ell-1,m)$. If $f \wedge x$ is decomposable, then $c_f$ is a minimum weight codeword of $C(\ell-1,m)$, and therefore $\wt(c_{f \wedge x})=q^{(\ell-1)(m-\ell+1)}$. Moreover, if $f \wedge x $ is non decomposable, then by induction hypothesis, $\wt(c_{f \wedge x})=q^{(\ell-1)(m-\ell+1)}+q^{(\ell-1)(m-\ell+1)-2}$. Consequently,
  
    \begin{align*}
        |Z| &=  (q^{m-\ell+2}-q^{m-\ell-2})\,q^{(\ell-1)(m-\ell+1)}\\
     &+(q^m-q^{m-\ell+2})(q^{(\ell-1)(m-\ell+1)}+q^{(\ell-1)(m-\ell+1)-2}) \\  
        = & q^{\ell-1}(q^\ell-1)\,q^{\ell(m-\ell)-2}(q^2+1).
    \end{align*}
  
    On the other hand, for each $L \in W(f), \phi^{-1}(L)$ consists of the pairs $(L',x)$ such that $L'$ is a hyperplane in $L$ and $x \in L \setminus L'$. Thus 
    $$|Z|=\sum_{L \in W(f)}\phi^{-1}(L)=\wt(c_f)\,q^{\ell-1}(q^\ell-1).$$
    Hence, we get that $\wt(c_f)=q^{\ell(m-\ell)-2}(q^2+1)$.
\end{proof}

Using the classification of the codewords of second minimum weight in the Grassmann code $C(\ell,m)$, we can now determine the number of codewords that attain the second minimum weight.

\begin{theorem}
    Let $2 \leq \ell \leq m-2$. The number of codewords of $C(\ell,m)$ having weight $q^{\ell(m-\ell)-2}(q^2+1)$ is
    $$
    A_{q^{\ell(m-\ell)-2}(q^2+1)} = \begin{bmatrix}
    m\\
    \ell +2
    \end{bmatrix}_q
    \begin{bmatrix}
    \ell+2\\
    4
\end{bmatrix}_q
    q^2(q-1)(q^3-1).
    $$
\end{theorem}
\begin{proof}
    Note that the evaluation map
    $$
    \Ev : \bigwedge^{m-\ell} V \longmapsto \Fq^{\,n} \quad \text{given by} \quad f \mapsto c_f
    $$
    is injective. In view of Theorem \ref{2a} and Theorem \ref{2b}, it is sufficient to count the number of $f \in \bigwedge^{m-\ell}V$ such that $f=v_1 \wedge \dots \wedge v_{m-\ell-2} \wedge \xi$ where $v_1, \dots, v_{m-\ell-2}$ are linearly independent and $\xi$ is an alternating form of rank 4. By Lemma \ref{d}, $V_f = \text{span}\{v_1, \dots, v_{m-\ell-2}\}$. Moreover, if $\{w_1, \dots, w_{m-\ell-2}\}$ is any other basis of $V_{f}$, then 
    $$
    v_1 \wedge \dots \wedge v_{m-\ell-2}=\lambda w_1 \wedge \dots w_{m-\ell-2} 
    \quad \text{where} \quad 
    \lambda \in \Fq ^{\times}.
    $$
    Thus, for a fixed subspace $U = V_f$, the change of basis of $U$ changes the exterior product $v_1 \wedge \dots \wedge v_{m-\ell-2}$ only up to a nonzero scalar.
    The number of $(m-\ell-2)$-dimensional subspace of $V$ is given by the Gaussian binomial $\begin{bmatrix} m \\ m-\ell-2 \end{bmatrix}_q = \begin{bmatrix} m \\ \ell +2 \end{bmatrix}_q$. Fix such a subspace U and choose an $(\ell+2)$ dimensional complement $W$ of $U$ in $V$, so that $V = U \oplus W$. We now need to count the rank four elements of $ \bigwedge ^2 W$, which is exactly equal to the number $N(\ell+2,4)$ of all bilinear skew-symmetric forms of rank 4. Therefore, the total number is
    $$
    \begin{bmatrix} m \\ \ell +2 \end{bmatrix}_q N(\ell+2,4) = \begin{bmatrix}
    m\\
    \ell +2
    \end{bmatrix}_q
    \begin{bmatrix}
    \ell+2\\
    4
\end{bmatrix}_q
    q^2(q-1)(q^3-1).
    $$
\end{proof}
\section*{Acknowledgment } The second author of this article is supported by the Indo-Norwegian Cooperation Programme in Higher Education and Research (INCP2) grant. He would like to thank the University Grant Commission of India for providing the financial support. He would also like to thank the Department of Mathematics and Statistics, UiT, Norway, for their warm hospitality during his visit to Tromso, where the work of this article started.

\bibliographystyle{plain}  
\bibliography{bib}   

\end{document}